\documentclass[a4paper,10pt,]{article}
\usepackage[utf8,]{inputenc}

\usepackage{setspace}
\usepackage{amsmath,amsfonts,amssymb}

\newtheorem{theorem}{Theorem}

\newtheorem{definition}[theorem]{Definition}
\newtheorem{remark}[theorem]{Remark}
\newtheorem{example}[theorem]{Example}
\newtheorem{corollary}[theorem]{Corollary}
\newtheorem{proposition}[theorem]{Proposition}
\newenvironment{proof}[1][Proof]{\noindent\textbf{#1: }}{\ \rule{0.5em}{0.5em}}

\title{Multilinear objective function-based clustering}
\author{Giovanni Rossi\\
\footnotesize{Departments of Computer Science and Engineering - DISI, and Mathematics}\\
\footnotesize{University of Bologna, Mura Anteo Zamboni 7, Italy 40126; \textsl{giovanni.rossi6@unibo.it}}}

\begin{document}

\maketitle

\begin{abstract}
The input of most clustering algorithms is a symmetric matrix quantifying similarity within data pairs. Such a matrix is here turned into a quadratic
set function measuring cluster score or similarity within data subsets larger than pairs. In general, any set function reasonably assigning a cluster score to
data subsets gives rise to an objective function-based clustering problem. When considered in pseudo-Boolean form, cluster score enables to evaluate fuzzy
clusters through multilinear extension MLE, while the global score of fuzzy clusterings simply is the sum over constituents fuzzy clusters of their MLE 
score. This is shown to be no greater than the global score of hard clusterings or partitions of the data set, thereby expanding a known result on extremizers
of pseudo-Boolean functions. Yet, a multilinear objective function allows to search for optimality in the interior of the hypercube. The proposed method only
requires a fuzzy clustering as initial candidate solution, for the appropriate number of clusters is implicitly extracted from the given data set.\smallskip

\footnotesize{\textsl{Keywords:} Fuzzy clustering, Similarity matrix, Pseudo-Boolean function, Multilinear extension, Gradient method, Local search.}
\end{abstract}

\section{Introduction}
Clustering means identifying groups within data, with the intent to have similarity between elements in a same group or cluster, and dissimilarity
between elements in different groups. It is important in a variety of fields at the interface of computer science, artificial intelligence and
engineering, including pattern recognition and learning, web mining and bioinformatics. In \textit{hard} clustering the sought group structure is
a partition of the data set, and clusters are blocks \cite{Aigner79}; each data point is in precisely one cluster, with full or unit membership. In
\textit{fuzzy} clustering data points distribute $[0,1]$-ranged memberships over clusters. This yields more flexibility, which is useful in many
applications \cite{Kashef+07,Valente+07}.

In \textit{objective function-based} clustering, the prevailing clusters obtain by maximizing or minimizing an objective function: any solution 
is mapped into a real quantity, i.e. its efficiency or cost, obtained as the sum over clusters of their own quality. Thus in hard clustering this
sum is over blocks, with a \textit{cluster score} set function taking real values over the $2^n-1$-set of non-empty data subsets \cite{Roubens82},
$n$ being the number of data. How to specify cluster score is the first issue addressed below.

When dealt with in pseudo-Boolean form, cluster score admits a unique polynomial multilinear extension MLE over $n$-dimensional unit hypercube
$[0,1]^n$. Such a MLE is a novel and seemingly appropriate measure of the score of fuzzy clusters. This is where to start for the clustering
method proposed here. As fuzzy clusterings are collections of fuzzy clusters, attention is placed on those such collections where the data
distribute memberships adding up to 1, with optimality found where the sum over fuzzy clusters of their MLE-score is maximal. If global cluster
score is evaluated via the MLE of cluster score, then its bounds are on hard clusterings. This expands a known result in pseudo-Boolean optimization
\cite{BorosHammer02}. Clustering is thus approached in terms of set partitioning, with this latter combinatorial optimization problem extended
from a discrete to a continuous domain \cite{Pardalos++06} and solved through a novel local search heuristic.
\subsection{Related work and approach}
Objective function-based fuzzy clustering mainly develops from the fuzzy $c$-means FcM \cite{Bezdek+92} and the possibilistic $c$-means PcM
\cite{Krishnapuram+96} algorithms. Given $n$ data points in $\mathbb R^m$, with $m$ observed features, both FcM and PcM iteratively act on
$c<n$ cluster centers or prototypes, aiming to minimize a cost: the sum over clusters of all distances between cluster elements and their center.
For any $c$ centers as input, at each iteration the memberships of data to clusters is re-defined so to minimize the sum of (fuzzy) distances from
centers, and next centers themselves are re-calculated so to minimize the sum of distances from (fuzzy) members. In FcM (but not in PcM) membership
distributions over the $c$ clusters add up to 1. The iteration stops when two consecutive fuzzy clusterings (specifying $c$ centers and $c\times n$
memberships) are sufficiently close (or coincide). This converges to a local minimum, and given non-convexity of the objective function, the choice
of suitable initial cluster centers is crucial. Initial cluster centers have to be arbitrary, and it is harsh to assess whether $c$ is a proper number
of clusters for the given data set \cite{Menard+02}. Much effort is thus devoted to finding the optimal number of clusters. One approach is to
validate fuzzy clusterings obtained at different values of $c$ by means of an index, and then selecting the value of $c$ for the output that scored
best on the index \cite{Zahid+++01}. This validation may be integrated into the FcM iterations, yielding a validity-guided method \cite{Bensaid++++++96}.
Main \textit{cluster validity indices} are: classification entropy, partition coefficient, uniform data functional, compactness and separation criteria.
They may be analyzed comparatively in terms membership distributions over clusters \cite{PalBezdek95}. A common idea is that clustering performance is
higher the more distributions are concentrated, as this formalizes a non-ambiguous classification \cite{Rezaee++98,XieBeni91}.

Recent clustering methods such as \textit{neural gas} \cite{HammerNeuralGas}, \textit{self organizing maps} \cite{WuChowSOM}, \textit{vector quantization}
\cite{Lughofer08,Du2010} and \textit{kernel methods} \cite{KernelMethods2004} maintain special attention on finding the optimal number of clusters for the
given data. In several classification tasks concerning protein structures, text documents, surveys or biological signals, an explicit metric vector space
(i.e. $\mathbb R^m$ above) is not available. Then, clustering may rely on \textit{spectral methods}, where the $\binom{n}{2}$ similarities within data
pairs are the adjacency matrix of a weighted graph, and the (non-zero) eigenvalues and eigenvectors of the associated Laplacian are used for partitioning
the data (or vertex) set. Specifically, the sought partition is to be such that lightest edges have endpoints in different blocks, whereas heaviest edges
have both endpoints in a same block. Although spectral clustering focuses on hard rather than fuzzy models, still it displays some analogy with the local
search method detailed below, as in both cases full membership of data points in prevailing clusters is decided in a single step. In fact, spectral methods
mostly focus on some first $c<n$ eigenvalues (in increasing order) \cite{SpectralClustering08,NgSpectral2002}, thus constraining the number of clusters. Here,
such a constraint is possible as well, although with suitable candidate solution the proposed local search autonomously finds an optimal (unrestricted) number
of clusters. 

Clustering is here approached by firstly quantifying the cluster score of every non-empty data subset, and secondly in terms of the associated set
partitioning combinatorial optimization problem \cite{KorteVygen2002}. Cluster score thus is a set function or, geometrically, a point in
$\mathbb R^{2^n-1}$, and rather than measuring a cost (or sum of distances) to be minimized (see above), it measures a worth to be maximized. The idea is to
quantify, for every data subset, both internal similarity and dissimilarity with respect to its complement. This resembles the
\textit{\textquotedblleft collective notion of similarity\textquotedblright} in information-based clustering \cite{PNAS2005}. Objective function-based
clustering intrisically relies on the assumption that every data subset has an associated real-valued worth (or, alternatively, a cost). A main novelty
proposed below is to deal with both hard and fuzzy clusters at once by means of the pseudo-Boolean form of set functions. In order to have the same input
as in many clustering algorithms, the basic cluster score function provided in the next section obtains from a given similarity matrix, and has polynomial
MLE of degree 2 \cite[pp. 157, 162]{BorosHammer02}. This also keeps the computational burden at a seemingly reasonable level.

\section{Cluster Score}
Given data set $N=\{1,\ldots ,n\}$, the input of most clustering algorithms \cite{XuWunsch05} is a symmetric similarity matrix $S\in [0,1]^{n\times n}$, 
with $S_{ij}$ quantifying similarity within data pairs $\{i,j\}\subset N$. If data points belong to a Euclidean space, i.e. $N\subset\mathbb R^m$, then
similarities $S_{ij}=1-d(i,j)$ may obtain through a normalized distance $d:N\times N\rightarrow[0,1]$. Not only pairs but also any data subset
$A\subseteq N$ may have a measurable internal similarity (and, possibly, dissimilarity with respect to its complement $A^c=N\backslash A$), interpreted
as its score $w(A)$ as a cluster. How to specify set function $w_S$ from given matrix $S$ is addressed hereafter.

For $2^N=\{A:A\subseteq N\}$, collection $\{\zeta(A,\cdot):A\in 2^N\}$ is a linear basis of the vector space $\mathbb R^{2^n}$ of real-valued functions $w$
on $2^N$, where $\zeta:2^N\times2^N\rightarrow\mathbb R$ is the element of the incidence algebra \cite{Rota64,Aigner79} of Boolean lattice
$(2^N,\cap,\cup)$ defined by $\zeta(A,B)=1$ if $B\supseteq A$ and $\zeta(A,B)=0$ if $B\not\supseteq A$. That is, the zeta function. Any $w$ corresponds to
linear combination\smallskip

$w(B)=\sum_{A\in 2^N}\mu^w(A)\zeta(A,B)=\sum_{A\subseteq B}\mu^w(A)$ for all $B\in 2^N$,\smallskip

with M\"obius inversion $\mu^w:2^N\rightarrow\mathbb R$ given by ($\subset$ is strict inclusion)\smallskip

$\mu^w(A)=\sum_{B\subseteq A}(-1)^{|A\backslash B|}w(B)\text{ }\left(\text{where }\zeta(B,A)=(-1)^{|A\backslash B|}\right)$, or\smallskip

$\mu^w(A)=w(A)-\sum_{B\subset A}\mu^w(B)\text{ }\left(\text{recursion, with }w(\emptyset)=0\right)$.\smallskip

This combinatorial \textit{\textquotedblleft analog of the fundamental theorem of the calculus\textquotedblright} \cite{Rota64} yields the unique MLE
$f^w:[0,1]^n\rightarrow\mathbb R$ of $w$, with values\smallskip

$w(B)=f^w(\chi_B)=\sum_{A\in 2^N}\left(\prod_{i\in A}\chi_B(i)\right)\mu^w(A)=\sum_{A\subseteq B}\mu^w(A)$ on vertices,\bigskip

and $f^w(q)=\sum_{A\in 2^N}\left(\prod_{i\in A}q_i\right)\mu^w(A)\textbf{ [1]}$\bigskip

on any $q=(q_1,\ldots ,q_n)\in[0,1]^n$. Conventionally, $\prod_{i\in\emptyset}q_i:=1$ \cite[p. 157]{BorosHammer02}.

A \textit{quadratic} MLE is a polynomial of degree 2: $\mu^w(A)=0$ if $|A|>2$. Geometrically, this means that $w$ is a point in a $\binom{n+1}{2}$-dimensional
vector (sub)space, i.e. $w\in\mathbb R^{\binom{n+1}{2}}$, as all its $2^n-1$ values are determined by the values taken by $\mu^w$ on the $n$ singletons and on the
$\binom{n}{2}$ pairs (with $n+\binom{n}{2}=\binom{n+1}{2}$). Similarity matrix $S$ factually has only $\binom{n}{2}$ valid entries (see below), and trying to
exploit them beyond a quadratic form for cluster score $w$ seems clumsy. Also, $S$ is intended precisely to measure such a score $S_{ij}$ for all
$\binom{n}{2}$ data pairs. The sought quadratic cluster score function $w$ is thus already defined on such pairs, i.e. $w(\{i,j\})=S_{ij}$ for all
$\{i,j\}\in 2^N$. How to assign scores $w(\{i\})$ to singletons $\{i\},i\in N$ seems a more delicate matter. If such scores are set equal to the $n$ entries
$S_{ii}=1$ along the main diagonal, then M\"obius inversion $\mu^w$ takes values $\mu^w(\{i\})=1$ on singletons and $\mu^w(\{i,j\})=S_{ij}-S_{ii}-S_{jj}<0$ on pairs
(while $\mu^w(A)=0$ for $1\neq|A|\neq 2$). This is a sufficient (but not necessary) condition for sub-additivity, i.e. $w(A\cup B)-w(A)-w(B)\leq 0$ for all
$A,B\in 2^N$ such that $A\cap B=\emptyset$. Then, the trivial partition where each data point is a singleton block is easily checked to be optimal. On the other hand,
setting $w(\{i\})=0$ for all $i\in N$ yields a M\"obius inversion with values $\mu^w(A)\geq 0$ for all $A\in 2^N$, and this is sufficient (but not necessary) for
super-additivity, i.e. $w(A\cup B)-w(A)-w(B)\geq 0$ for all $A,B\in 2^N$ with $A\cap B=\emptyset$. Then, it becomes optimal (and again trivial) to place all $n$ data
points into a unique grand cluster. An seemingly appropriate alternative to these two unreasonable situations is\smallskip

$w(\{i\})=\frac{1}{2}-\frac{1}{2(n-1)}\sum_{l\in N\backslash i}S_{il}=\sum_{l\in N\backslash i}\frac{1-S_{il}}{2(n-1)}$.\smallskip

In this way, $1-S_{il}$ quantifies diversity between $i\in N$ and $l\in N\backslash i$, which must be equally shared among them. The cluster score of
singleton $\{i\}$ then is the average of such $n-1$ diversities $\frac{1-S_{il}}{2},l\in N\backslash i$ collected by $i$. M\"obius inversion takes values
$\mu^w(\{i\})=w(\{i\})$ on singletons and, recursively,\smallskip

$\mu^w(\{i,j\})=\frac{nS_{ij}-1}{n-1}-\sum_{l\in N\backslash\{i,j\}}\frac{2-(S_{il}+S_{jl})}{2(n-1)}$\smallskip

on pairs. Note that the cluster score $w(A)$ of any $A\in 2^N$ does not depend on those $\binom{n-a}{2}$ entries $S_{ll'}$ of matrix $S$ such that $l,l'\in A^c$,
where $a=|A|$.

In particular, if $S_{ij}=1$ for $i\in A,j\in A\backslash i$ and $S_{il}=0$ for $i\in A,l\in A^c$, then\smallskip

$w(A)=\frac{a(a^2-3a+n+1)}{2(n-1)}=a\cdot\frac{n-a}{2(n-1)}+\binom{a}{2}\cdot\left(1-\frac{n-a}{n-1}\right)$,\smallskip

with $w(A)=\frac{1}{2},1,\ldots ,\frac{n^2-4n+5}{2},\binom{n}{2}$ for $a=1,2,\ldots ,n-1,n$. In terms of spectral clustering (see above), the corresponding adjacency matrix
identifies a subgraph spanned by vertex subset $A$ which is a clique, and where all edges have maximum weight, i.e. 1, with remarkable implications for the eigenvalues
of the normalized Laplacian; see \cite[p. 42]{GraphClustering07}.

This quadratic $w$, obtained from given similarity matrix $S$, is conceived as the
main example of a cluster score function. Multilinear objective function-based clustering deals with generic set functions, possibly non-quadratic, but reasonably measuring
a cluster score of data subsets. In fact, the MLE $\textbf{[1]}$ above, of set functions $w$, is where to start for the investigation proposed in the remainder of this work.
In view of the rich literature on pseudo-Boolean methods \cite{BooleanFunctions}, the MLE of cluster score appears very suitable for evaluating fuzzy clusters, especially in
that established definitions of neighborhood and derivative may be expanded to fit the broader setting formalized in the sequel. Specifically, while the $n$ variables of traditional
pseudo-Boolean functions range each in $[0,1]$, the $n$ variables of the novel near-Boolean form \cite{Rossi2015} considered here range each in a $2^{n-1}-1$-dimensional unit
simplex. The reason for this is the purpose to use the MLE of cluster score for evaluating not only fuzzy clusters, but also and most importantly fuzzy clusterings, which are
collections $q^1,\ldots ,q^m\in[0,1]^n$ of fuzzy clusters, and thus global score is quantifiable as $\sum_{1\leq k\leq m}f^w(q^k)$. In this respect, note that PcM
algorithms allow for membership distributions adding up to quantities $<1$ (see above) for handling outliers. These shall be placed each in a singleton block of the partition found
here via gradient-based local search. Therefore, membership distributions are like in FcM methods, i.e. ranging in a $2^{n-1}-1$-dimensional unit simplex. 

\section{Fuzzy Clustering}
A fuzzy clustering is a $m$-set $q^1,\ldots ,q^m\in [0,1]^n$ of fuzzy clusters or points in the $n$-dimensional unit hypercube, with $(q_i^1,\ldots ,q_i^m)$ being $i$'s membership
distribution. The $2^{n-1}$-set $2^N_i=\{A:i\in A\in 2^N\}$ of subsets containing each $i\in N$ has associated $2^{n-1}-1$-dimensional unit simplex\bigskip

$\Delta_i=\left\{\left(q_i^{A_1},\ldots ,q_i^{A_{2^{n-1}}}\right):q_i^{A_k}\geq 0\text{ for }1\leq k\leq 2^{n-1},\sum_{1\leq k\leq 2^{n-1}}q_i^{A_k}=1\right\}$,\bigskip

where $\left\{A_1,\ldots ,A_{2^{n-1}}\right\}=2^N_i$ and $q_i\in\Delta_i$ is $i$'s membership distribution.
\begin{definition}
A fuzzy cover $\textbf q$ specifies, for each data point $i\in N$, a membership distribution over the $2^{n-1}$ data subsets $A\in 2^N_i$ containing it. Hence
$\textbf q=(q_1,\ldots ,q_n)\in\Delta_N$, where $\Delta_N=\times_{1\leq i\leq n}\text{ }\Delta_i$.
\end{definition}
Equivalently, $\textbf q=\left\{q^A:\emptyset\neq A\in 2^N,q^A\in[0,1]^n\right\}$ is a $2^n-1$-set  whose elements $q^A=\left(q^A_1,\ldots ,q^A_n\right)$ are points in the
$n$-dimensional hypercube corresponding to non-empty data subsets $\emptyset\neq A\in 2^N$ and specifying a membership $q_i^A$ for each $i\in N$, with $q_i^A\in[0,1]$ if $i\in A$
while $q_j^A=0$ if $j\in A^c$. Fuzzy covers thus generalize traditional fuzzy clusterings, as these latter are commonly intended as collections $\{q^1,\ldots ,q^m\}$ as above where,
in addition, for every fuzzy cluster $q^k$ the associated data subset implicitly is $\{i:q^k_i>0\}$. Conversely, fuzzy covers allow for situations where $0<|\{i:q_i^A>0\}|<|A|$ for
some $\emptyset\neq A\in 2^N$, although an \textit{exactness} condition introduced below shows that such cases may be ignored.

Fuzzy covers being collections of points in $[0,1]^n$, and the MLE $f^w$ of $w$ allowing precisely to evaluate such points, the global score $W(\textbf q)$ of any
$\textbf q\in\Delta_N$ is the sum over all its elements $q^A,A\in 2^N$ of their own score as quantified by $f^w$ (see $\textbf{[1]}$). That is,\bigskip

$W(\textbf q)=\sum_{A\in 2^N}f^w(q^A)=\sum_{A\in 2^N}\left[\sum_{B\subseteq A}\left(\prod_{i\in B}q_i^A\right)\mu^w(B)\right]$,\bigskip

or $W(\textbf q)=\sum_{A\in 2^N}\left[\sum_{B\supseteq A}\left(\prod_{i\in A}q_i^B\right)\right ]\mu^w(A)\textbf{ [2]}$.\bigskip

\begin{example}
For $N=\{1,2,3\}$, consider $w(\{1\})=w(\{2\})=w(\{3\})=0.2$, $w(\{1,2\})=0.8$, $w(\{1,3\})=0.3$, $w(\{2,3\})=0.6$, $w(N)=0.7$. Membership distributions over subsets $2^N_i,i=1,2,3$ are
$q_1\in\Delta_1,q_2\in\Delta_2,q_3\in\Delta_3$,\smallskip

$q_1=\left(\begin{array}{c}q_1^1 \\q_1^{12} \\q_1^{13} \\q_1^N\end{array}\right)\text{, }
q_2=\left(\begin{array}{c}q_2^2 \\q_2^{12} \\q_2^{23} \\q_2^N\end{array}\right)\text{, } 
q_3=\left(\begin{array}{c}q_3^3 \\q_3^{13} \\q_3^{23} \\q_3^N\end{array}\right)$.\bigskip

If $\hat q_1^{12}=\hat q_2^{12}=1$, then any membership $q_3\in\Delta_3$ yields\smallskip

$w(\{1,2\})+\left(q_3^3+q_3^{13}+q_3^{23}+q_3^N\right)\mu^w(\{3\})=w(\{1,2\})+w(\{3\})=1$.\smallskip

This means that there is a continuum of fuzzy covers achieving maximum score: $W(\hat q_1,\hat q_2,q_3)=1$ independently from $q_3$. In order to select 
$\hat{\textbf q}=(\hat q_1,\hat q_2,\hat q_3)$ where $\hat q_3^3=1$, attention must be placed only on \textit{exact} ones, defined hereafter.
\end{example}
For any two fuzzy covers $\textbf q=\{q^A:\emptyset\neq A\in 2^N\}$ and $\hat{\textbf q}=\{\hat q^A:\emptyset\neq A\in 2^N\}$, define $\hat{\textbf q}$ to be a \textit{shrinking} of
$\textbf q$ if $A\in 2^N,|A|>1$ satisfies $\sum_{i\in A}q_i^A>0$ and\bigskip

$\hat q^B_i=q^B_i\text{ if } B\not\subseteq A$ and $\hat q^B_i= 0\text{ if }B=A,\text{ for all }B\in 2^N,i\in N$,\bigskip

$\sum_{B\subset A}\hat q^B_i=q^A_i+\sum_{B\subset A}q_i^B\text{ for all }i\in A$.\bigskip

In words, a shrinking reallocates the whole membership mass $\sum_{i\in A}q_i^A>0$ from $A\in 2^N$ to all proper subsets $B\subset A$, involving all and only
those data points $i\in A$ with strictly positive membership $q_i^A>0$.
\begin{definition}
Fuzzy cover $\textbf q\in\Delta_N$ is exact as long as $W(\textbf q)\neq W(\hat{\textbf q})$ for all shrinkings $\hat{\textbf q}$ of \textbf q.
\end{definition}
\begin{proposition}
If $\textbf q$ is exact, then $\left|\left\{i\in A:q_i^A>0\right\}\right|\in\{0,|A|\}$ for all $A\in 2^N$.
\end{proposition}
\begin{proof}
For $\emptyset\subset A^+(\textbf q)=\left\{i:q_i^A>0\right\}\subset A$, with $\alpha=|A^+(\textbf q)|>1$, note that\bigskip

$f^w(q^A)=\sum_{B\subseteq A^+(\textbf q)}\left(\prod_{i\in B}q_i^A\right)\mu^w(B)$.\bigskip

Let shrinking $\hat{\textbf q}$, with $\hat q^{B'}=q^{B'}$ if $B'\not\in 2^{A^+(\textbf q)}$, satisfy conditions\bigskip

1) $\sum_{B\in 2^N_i\cap 2^{A^+(\textbf q)}}\hat q_i^B=q_i^A+\sum_{B\in 2^N_i\cap 2^{A^+(\textbf q)}}q_i^B\text{ for all }i\in A^+(\textbf q)$, and\bigskip

2) $\prod_{i\in B}\hat q_i^B=\prod_{i\in B}q_i^B+\prod_{i\in B}q_i^A\text{ for all }B\in 2^{A^+(\textbf q)}$ such that $|B|>1$.\bigskip

These are $2^\alpha-1$ equations with $\sum_{1\leq k\leq\alpha}k\binom{\alpha}{k}>2^{\alpha}$ variables $\hat q_i^B,B\subseteq A^+(\textbf q)$, $i\in B$. Thus there is a
continuum of solutions, each providing precisely a shrinking $\hat{\textbf q}$ where\bigskip

$\sum_{B\in 2^{A^+(\textbf q)}}f^w(\hat q^B)=f^w(q^A)+\sum_{B\in 2^{A^+(\textbf q)}}f^w(q^B)$.\bigskip

This entails that \textbf q is not exact.
\end{proof}

For given $w$, the global score of any fuzzy cover also attains on many fuzzy clusterings. This justifies the following (in line with standard terminology). 
\begin{definition}
Fuzzy clusterings are exact covers.
\end{definition}
The global score of any fuzzy clustering is shown below to also attain on some hard clustering, thereby expanding a result on extremizers of pseudo-Boolean
functions \cite[p. 163]{BorosHammer02}. 

\section{Hard clustering}
\label{sec:hard-clusterings}
Hard clusterings or partitions of $N$ \cite{Aigner79} are fuzzy clusterings where $q_i^A\in\{0,1\}$ for all $A\in 2^N$ and all $i\in A$. Among the maximizers of any
objective function $W:\Delta_N\rightarrow\mathbb R$ as above there always exist fuzzy clusterings $(q_1,\ldots ,q_n)\in\Delta_N$ such that $q_i\in ex(\Delta_i)$ for
all $i\in N$, where $ex(\Delta_i)$ denotes the $2^{n-1}$-set of extreme points of $\Delta_i$. For $\textbf q\in\Delta_N,i\in N$, let $\textbf q=q_i|\textbf q_{-i}$,
with $q_i\in\Delta_i$ and $\textbf q_{-i}\in\Delta_{N\backslash i}=\times_{j\in N\backslash i}\text{ }\Delta_j$. Then, for any $w$,\bigskip

$W(\textbf q)=\sum_{A\in 2^N_i}f^w(q^A)+\sum_{A'\in 2^N\backslash 2^N_i}f^w(q^{A'})=$\bigskip

$=\sum_{A\in2^N_i}\sum_{B\subseteq A\backslash i}\left(\prod_{j\in B}q^A_j\right)\Big(q^A_i\mu^w(B\cup i)+\mu^w(B)\Big)+$\bigskip

$+\sum_{A'\in 2^N\backslash 2^N_i}\sum_{B'\subseteq A'}\left(\prod_{j'\in B'}q^{A'}_{j'}\right)\mu^w(B')$\bigskip

at all $\textbf q\in\Delta_N$ and for all $i\in N$. Now define\bigskip

$W_i(q_i|\textbf q_{-i})=\sum_{A\in 2^N_i}q^A_i\left[\sum_{B\subseteq A\backslash i}\left(\prod_{j\in B}q^A_j\right)\mu^w(B\cup i)\right]$,\bigskip

$W_{-i}(\textbf q_{-i})=\sum_{A\in 2^N_i}\left[\sum_{B\subseteq A\backslash i}\left(\prod_{j\in B}q^A_j\right)\mu^w(B)\right]+$\bigskip

$+\sum_{A'\in 2^N\backslash 2^N_i}\left[\sum_{B'\subseteq A'}\left(\prod_{j'\in B'}q^{A'}_{j'}\right)\mu^w(B')\right]$,\bigskip

$\text{yielding }W(\textbf q)=W_i(q_i|\textbf q_{-i})+W_{-i}(\textbf q_{-i})\textbf{ [3]}$.

\begin{proposition}
For all $\textbf q\in\Delta_N$, there are $\underline{\textbf q},\overline{\textbf q}\in\Delta_N$ such that

(i) $\text{(i) }W(\underline{\textbf q})\leq W(\textbf q)\leq W(\overline{\textbf q})$ and,

(ii) $\underline q_i,\overline q^i\in ex(\Delta_i)\text{ for all }i\in N$.
\end{proposition}
\begin{proof}
For all $i\in N$ and $\textbf q_{-i}\in\Delta_{N\backslash i}$, define
$w_{\textbf q_{-i}}:2^N_i\rightarrow\mathbb R$ by\bigskip

$w_{\textbf q_{-i}}(A)=\sum_{B\subseteq A\backslash i}\left(\prod_{j\in B}q^A_j\right)\mu^w(B\cup i)\textbf{ [4]}$.

Let $\mathbb A^+_{\textbf q_{-i}}=\arg\max w_{\textbf q_{-i}}$ and $\mathbb A^-_{\textbf q_{-i}}=\arg\min w_{\textbf q_{-i}}$,
with $\mathbb A^+_{\textbf q_{-i}}\neq\emptyset\neq\mathbb A^-_{\textbf q_{-i}}$ at all $\textbf q_{-i}$. Most importantly,\bigskip

$W_i(q_i|\textbf q_i)=\sum_{A\in 2^N_i}\Big(q^A_i\cdot w_{\textbf q_{-i}}(A)\Big)=\langle q_i,w_{\textbf q_{-i}}\rangle\textbf{ [5]}$,\bigskip

where $\langle\cdot ,\cdot\rangle$ denotes scalar product. Thus for given membership distributions of all $j\in N\backslash i$, global score is affected by $i$'s
membership distribution through a scalar product. In order to maximize (or minimize) $W$ by suitably choosing $q_i$ for given $\textbf q_{-i}$, the whole of $i$'s membership mass
must be placed over $\mathbb A^+_{\textbf q_{-i}}$ (or $\mathbb A^-_{\textbf q_{-i}}$), anyhow. Hence there are precisely $|\mathbb A^+_{\textbf q_{-i}}|>0$
(or $|\mathbb A^-_{\textbf q_{-i}}|>0$) available extreme points of $\Delta_i$. The following procedure selects (arbitrarily) one of them.\smallskip

\textsc{RoundUp}$(w,\textbf q)$\smallskip

\textsl{Initialize:} Set $t=0$ and $\textbf q(0)=\textbf q$.\smallskip

\textsl{Loop:} While there is a $i\in N$ with $q_i(t)\not\in ex(\Delta_i)$,\smallskip

set $t=t+1$ and:
\begin{enumerate}
\item[(a)] select some $A^*\in\mathbb A^+_{\textbf q_{-i}(t)}$,
\item[(b)] define, for all $j\in N,A\in 2^N$,\smallskip

$q^A_j(t)=q_j^A(t-1)\text{ if }j\neq i$,\smallskip

$q^A_j(t)=1\text{ if }j=i\text{ and } A=A^*$,\smallskip

$q^A_j(t)=0$ otherwise.
\end{enumerate}

\textsl{Output:} Set $\overline{\textbf q}=\textbf q(t)$.\smallskip

Every change $q_i^A(t-1)\neq q_i^A(t)=1$ (for any $i\in N,A\in 2^N_i$) induces a non-decreasing variation
$W(\textbf q(t))-W(\textbf q(t-1))\geq 0$. Hence, the sought $\overline{\textbf q}$ is provided in at most $n$ iterations.
Analogously, replacing $\mathbb A^+_{\textbf q_{-i}}$ with $\mathbb A^-_{\textbf q_{-i}}$ yields the sought minimizer $\underline{\textbf q}$.
\end{proof}
\begin{remark}
For $i\in N,A\in 2^N_i$, if all $j\in A\backslash i\neq\emptyset$ satisfy $q_j^A=1$, then $\textbf{[4]}$ yields $w_{\textbf q_{-i}}(A)=w(A)-w(A\backslash i)$, while 
$w_{\textbf q_{-i}}(\{i\})=w(\{i\})$ regardless of $\textbf q_{-i}$. For quadratic $w$ obtained above from similarity matrix $S$,\bigskip

$w_{\textbf q_{-i}}(A)=w(\{i\})+\sum_{j\in A\backslash i}q_j^A\text{ }\mu^w(\{i,j\})$.
\end{remark}
If the global score of fuzzy clusterings is quantified as the sum over constituents fuzzy clusters of their MLE-score, then for any $w$ there are hard clusterings among
both the maximizers and minimizers. This seems crucial because many applications may be modeled in terms of \textit{set partitioning}, and in such a combinatorial
optimization problem fuzzy clustering is not feasible. An important example is winner determination in combinatorial auctions \cite{Sandholm02}, where a set $N$ of items
to be sold must be partitioned into bundles towards revenue maximization. The maximum bid received for each bundle $\emptyset\neq A\subseteq N$ defines the input set
function $w$. The above result entails that if the objective function is multilinearly extended over the continuous domain of fuzzy clusterings, then any found solution
can be promptly adapted to the restricted domain of partitions, with no score loss. The problem can thus be approached from a geometrical perspective, allowing for novel
search strategies. Partitions $P=\{A_1,\ldots ,A_{|P|}\}\subset 2^N$ of $N$ are families of pairwise disjoint subsets whose union is $N$, i.e.
$N=\cup_{1\leq k\leq |P|}\text{ }A_k$ and $A_k\cap A_l=\emptyset,1\leq k<l\leq |P|$. Any $P$ corresponds to the collection $\{\chi_A:A\in P\}$ of those $|P|$ hypercube
vertices identified by the characteristic functions of its blocks (see above). Partitions $P$ can also be seen as $\textbf p\in\Delta_N$ where $p_i^A=1$ for all
$A\in P,i\in A$. The above findings yield the following.
\begin{corollary}
For any $w$, some partition $P$ satisfies $W(\textbf p)\geq W(\textbf q)$\smallskip

for all $\textbf q\in\Delta_N$, with $W(\textbf p)=\sum_{A\in P}w(A)$.
\end{corollary}
\begin{proof}
Follows from propositions 4 and 6.
\end{proof}

A further remark concerns cluster validity \cite{WangZhang2007}, with focus on those indices that validate fuzzy clusterings by relying exclusively on membership distributions.
As already observed, a basic argument is that the more such distributions are concentrated, the less ambiguous is the fuzzy classification. Evidently, hard clusterings provide
$n$ distributions each concentrated on a unique extreme point of the associated unit simplex. The above result indicates that if global score is evaluated through MLE, then
validation may ignore membership distributions, as the score of any optimal fuzzy clustering also obtains by means of a hard one.

\section{Local search}
\label{sec:search}
Defining global maximizers is clearly immediate.
\begin{definition}
Fuzzy clustering $\hat{\textbf q}\in\Delta_N$ is a global maximizer if $W(\hat{\textbf q})\geq W(\textbf q)$ for all $\textbf q\in\Delta_N$. 
\end{definition}
Concerning local maximizers, consider a vector $\omega=(\omega_1,\ldots ,\omega_n)\in\mathbb R^n_{++}$ of strictly positive weights, with $\omega_N=\sum_{j\in N}\omega_j$,
and focus on the equilibrium \cite{Micro} of the game where data points are players who strategically choose their memberships distribution $q_i\in\Delta_i$ with payoff
equal to fraction $\frac{\omega_i}{\omega_N}W(q_1,\ldots ,q_n)$ of the global score attained at any strategy profile $(q_1,\ldots ,q_n)$.
\begin{definition}
Fuzzy clustering $\hat{\textbf q}\in\Delta_N$ is a local maximizer if

$W_i(\hat q_i|\hat{\textbf q}_{-i})\geq W_i(q_i|\hat{\textbf q}_{-i})$ for all $q_i\in\Delta_i$ and all $i\in N$ (see $\textbf{[3]}$).
\end{definition}
This definition of local maximizer entails that the \textit{neighborhood} $\mathcal N(\textbf q)\subset\Delta_N$ of any $\textbf q\in\Delta_N$ is\smallskip

$\mathcal N(\textbf q)=\bigcup_{i\in N}\Big\{\tilde{\textbf q}:\tilde{\textbf q}=\tilde q_i|\textbf q_{-i},\tilde q_i\in\Delta_i\Big\}$.

\begin{definition}
The $(i,A)$-derivative of $W$ at $\textbf q\in\Delta_N$ is\bigskip

$\partial W(\textbf q)/\partial q^A_i=W(\overline{\textbf q}(i,A))-W(\underline{\textbf q}(i,A))=$\bigskip

$=W_i\Big(\overline q_i(i,A)|\overline{\textbf q}_{-i}(i,A)\Big)-W_i\Big(\underline q_i(i,A)|\underline{\textbf q}_{-i}(i,A)\Big)$,\bigskip

with $\overline{\textbf q}(i,A)=\Big(\overline q_1(i,A),\ldots ,\overline q_n(i,A)\Big)$ given by\bigskip

$\overline q_j^B(i,A)=q_j^B\text{ for all }j\in N\backslash i,B\in 2^N_j$\smallskip

$\overline q_j^B(i,A)=1\text{ for }j=i,B=A$ (hence $\overline q_j^B(i,A)=0\text{ for }j=i,B\neq A$)\bigskip

and $\underline{\textbf q}(i,A)=\Big(\underline q_1(i,A),\ldots ,\underline q_n(i,A)\Big)$ given by\bigskip

$q_j^B(i,A)=q_j^B\text{ for all }j\in N\backslash i,B\in 2^N_j$\smallskip

$q_j^B(i,A)=0\text{ for }j=i\text{ and all }B\in 2^N_i$\bigskip

thus $\nabla W(\textbf q)=\{\partial W(\textbf q)/\partial q^A_i:i\in N,A\in 2^N_i\}\in\mathbb R^{n2^{n-1}}$ is the (full) gradient of $W$ at $\textbf q$.
The $i$-gradient $\nabla_iW(\textbf q)\in\mathbb R^{2^{n-1}}$ of $W$ at $\textbf q=q_i|\textbf q_{-i}$ is set function $\nabla_iW(\textbf q):2^N_i\rightarrow\mathbb R$ defined by
$\nabla_iW(\textbf q)(A)=w_{\textbf q_{-i}}(A)$ for all $A\in 2^N_i$, where $w_{\textbf q_{-i}}$ is given by $\textbf{[4]}$.
\end{definition}
\begin{remark}
Membership distribution $\underline q_i(i,A)$ is the null one: its $2^{n-1}$ entries are all 0, hence $\underline q_i(i,A)\not\in\Delta_i$.
\end{remark}
The setting obtained thus far allows to conceive searching for a local maximizer hard clustering $\textbf q^*$ from given fuzzy clustering \textbf q as initial
candidate solution, and while maintaing the whole search within the continuum of fuzzy clusterings. This idea may be specified in alternative ways
yielding different local search methods. One possibility is the following.\smallskip

\textsc{LocalSearch}$(w,\textbf q)$\smallskip

\textsl{Initialize:} Set $t=0$ and $\textbf q(0)=\textbf q$, with requirement $|\{i:q_i^A>0\}|\in\{0,|A|\}$ for all $A\in 2^N$.\smallskip

\textsl{Loop 1:} While $0<\sum_{i\in A}q^A_i(t)<|A|$ for a $A\in 2^N$, set $t=t+1$ and
\begin{enumerate}
\item[(a)] select a $A^*(t)\in 2^N$ such that\bigskip

$\sum_{i\in A^*(t)}w_{\textbf q_{-i}(t-1)}(A^*(t))\geq\sum_{j\in B}w_{\textbf q_{-j}(t-1)}(B)$\bigskip

for all $B\in 2^N$ such that $0<\sum_{i\in B}q^B_j(t)<|B|$,
\item[(b)] for $i\in A^*(t)$ and $A\in 2^N_i$, define $q_i^A(t)=1\text{ if }A=A^*(t)$, and\smallskip

$q_i^A(t)=0\text{ if }A\neq A^*(t)$;

\item[(c)] for $j\in N\backslash A^*(t)$ and $A\in 2^N_j$ with $A\cap A^*(t)=\emptyset$, define $q^A_j(t)=q_j^A(t-1)+$\smallskip

$+\left(w(A)\sum_{B\cap A^*(t)\neq\emptyset:B\in 2^N_j}q_j^B(t-1)\right)
\left(\sum_{B'\cap A^*(t)=\emptyset:B'\in 2^N_j}w(B')\right)^{-1}$,

\item[(d)] for $j\in N\backslash A^*(t)$ and $A\in 2^N_j$ with $A\cap A^*(t)\neq\emptyset$, define $q^A_j(t)=0$.

\end{enumerate}
\textsl{Loop 2:} While $q_i^A(t)=1,|A|>1$ for a $i\in N$ and $w(A)<w(\{i\})+w(A\backslash i)$, set $t=t+1$ and define:\bigskip

$q^{\hat A}_i(t)=1\text{ if }|\hat A|=1$ for all $\hat A\in 2^N_i$;\bigskip

$q^{B}_j(t)=1\text{ if }B=A\backslash i\text{ for all }j\in A\backslash i,B\in 2^N_j$;\bigskip

$q^{\hat B}_{j'}(t)=q^{\hat B}_{j'}(t-1)\text{ for all }j'\in A^c,\hat B\in 2^N_{j'}$.\bigskip

\textsl{Output:} Set $\textbf q^*=\textbf q(t)$.\smallskip

Both \textsc{RoundUp} and \textsc{LocalSearch} yield a sequence $\textbf q(0),\ldots ,\textbf q(t^*)=\textbf q^*$ where $q_i^*\in ex(\Delta_i)$ for all $i\in N$.
In the former at the end of each iteration $t$ the novel $\textbf q(t)\in\mathcal N(\textbf q(t-1))$ is in the neighborhood of its predecessor. In the latter
$\textbf q(t)\not\in\mathcal N(\textbf q(t-1))$ in general, as in $|P|\leq n$ iterations of \textsl{Loop 1} a partition $\{A^*(1),\ldots ,A^*(|P|)\}=P$ is generated.
Selected clusters or blocks $A^*(t)\in 2^N$, $t=1,\ldots ,|P|$ are any of those where the sum over data points $i\in A^*(t)$ of $(i,A^*(t))$-derivatives
$\partial W(\textbf q(t-1))/\partial q^{A^*(t)}_i(t-1)$ is maximal. Once a block $A^*(t)$ is selected, then lines (c) and (d) make all data points
$j\in N\backslash A^*(t)$ redistribute the entire membership mass currently placed on subsets $A'\in 2^N_j$ with non-empty intersection $A'\cap A^*(t)\neq\emptyset$
over those remaining $A\in 2^N_j$ such that, conversely, $A\cap A^*(t)=\emptyset$. The redistribution is such that each of these latter gets a fraction
$w(A)/\sum_{B\in 2^N_j:B\cap A^*(t)=\emptyset}w(B)$ of the newly freed membership mass $\sum_{A'\in 2^N_j:A'\cap A^*(t)\neq\emptyset}q_j^{A'}(t-1)$.  
The subsequent \textsl{Loop 2} checks whether the partition generated by \textsl{Loop 1} may be improved by exctracting some outliers from existing blocks and putting
them in singleton blocks of the final output. An outlier basically is a data point displaying very unusual features. In the limit, cluster score $w$ may be such that
for some data points $i\in N$ global score decreases when $i$ joins any cluster $A\in 2^N_i,|A|>1$, that is to say\smallskip

$w(A)-w(A\backslash i)-w(\{i\})=\sum_{B\in 2^A\backslash 2^{A\backslash i}:|B|>1}\mu^w(B)<0$.
\begin{proposition}
Output $\textbf q^*$ of \textsc{LocalSearch}$(W,\textbf q)$ is a local maximizer.
\end{proposition}
\begin{proof}
It is plain that the output corresponds to a partition $P$. With the notation of corollary 8 in section 4, $\textbf q^*=\textbf p$. Accordingly, any data point $i\in N$ is
either in a singleton cluster $\{i\}\in P$ or else in a cluster $A\in P,i\in A$ such that $|A|>1$. In the former case, any membership reallocation deviating from $p_i^{\{i\}}=1$,
given memberships $p_j,j\in N\backslash i$, yields a cover (fuzzy or hard) where global score is the same as at $\textbf p$, since $\prod_{j\in B\backslash i}p_j^B=0$ for all
$B\in 2_i^N\backslash A$ (see example 2 above). In the latter case, any membership reallocation $q_i$ deviating from $p_i^A=1$ (given memberhips $p_j,j\in N\backslash i$) yields
a cover which is best seen by distinguishing between $2_i^N\backslash A$ and $A$. Also recall that $w(A)-w(A\backslash i)=\sum_{B\in 2^A\backslash 2^{A\backslash i}}\mu^w(B)$.
Again, all membership mass $\sum_{B\in 2_i^N\backslash A}q_i^B>0$ simply collapses on singleton $\{i\}$ because $\prod_{j\in B\backslash i}p_j^B=0$ for all $B\in 2_i^N\backslash A$.
Therefore,\smallskip

$W(\textbf p)-W(q_i|\textbf p_{-i})=w(A)-w(\{i\})+$\bigskip

$-\left(q_i^A\sum_{B\in 2^A\backslash 2^{A\backslash i}:|B|>1}\mu^w(B)+\sum_{B'\in 2^{A\backslash i}}\mu^w(B')\right)=$\bigskip

$=\left(p_i^A-q_i^A\right)\sum_{B\in 2^A\backslash 2^{A\backslash i}:|B|>1}\mu^w(B)$.\bigskip

Now assume that \textbf q is \textit{not} a local maximizer, i.e. $W(\textbf p)-W(q_i|\textbf p_{-i})<0$. Since $p_i^A-q_i^A>0$ (because $p_i^A=1$ and $q_i\in\Delta_i$ is a deviation
from $p_i$), then\bigskip

$\sum_{B\in 2^A\backslash 2^{A\backslash i}:|B|>1}\mu^w(B)=w(A)-w(A\backslash i)-w(\{i\})<0$.\bigskip

Hence \textbf q cannot be the output of \textsl{Second Loop}.
\end{proof}

In local search methods, the chosen initial canditate solution determines what neighborhoods shall be visited. The range of the objective function in a neighborhood is a set of
real values. In a neighborhood $\mathcal N(\textbf p)$ of a hard clustering $\textbf p$ or partition $P$ only those $\sum_{A\in P:|A|>1}|A|$ data points $i\in A$ in non-sigleton
blocks $A\in P,|A|>1$ can modify global score by reallocating their membership. In view of the above proof, the only admissible variations obtain by deviating from
$p_i^A=1$ with an alternative membership distribution $q_i$ such that $q_i^A\in[0,1)$, with\bigskip

$W(q_i|\textbf p_{-i})-W(\textbf p)=(q_i^A-1)\sum_{B\in 2^A\backslash 2^{A\backslash i},|B|>1}\mu^w(B)+(1-q_i^A)w(\{i\})$.\bigskip

Hence, choosing partitions as initial candidate solutions of \textsc{LocalSearch} is evidently
poor. A sensible choice should conversely allow the search to explore different neighborhoods where the objective function may range widely. A simplest example of such an initial
candidate solution is uniform distribution $q_i^A=2^{1-n}$. On the other hand, the input of local search fuzzy clustering algorithms is commonly desired to be close to a global
optimum, i.e. a maximizer in the present setting. This translates here into the idea of defining a suitable input by means of cluster score function $w$. Along this line, consider
$q_i^A=w(A)/\sum_{B\in 2^N_i}w(B)$, yielding $\frac{q_i^A}{q_i^B}=\frac{w(A)}{w(B)}=\frac{q^A_j}{q^B_j}$ for all $A,B\in 2^N_i\cap 2^N_j$ and all $i,j\in N$.

With a suitable initial candidate solution, the search may be restricted to explore only a maximum number of fuzzy clusterings, thereby containing (together with the quadratic MLE of
cluster score $w$) the computational burden. In particular, if $\textbf q(0)$ is the finest partition $\{\{1\},\ldots ,\{n\}\}$ or $q_i^{\{i\}}(0)=1$ for all $i\in N$, then the
search does not explore any neighborhood at all, and such an input coincides with the output. More reasonably, let $\mathbb A_{\textbf q}^{\max}=\{A_1,\ldots ,A_k\}$ denote the
collection of maximal data subsets where input memberships are strictly positive. That is, $q_i^{A_{k'}}>0$ for all $i\in A_{k'},1\leq k'\leq k$ as well as $q_j^B=0$ for all
$B\in 2^N\backslash\left(2^{A_1}\cup\cdots\cup2^{A_k}\right)$ and all $j\in B$. Then, the output shall be a partition $P$ each of whose blocks $A\in P$ satisfies $A\subseteq A_{k'}$
for some $1\leq k'\leq k$. Hence, by suitably choosing the input \textbf q, \textsc{LocalSearch} outputs a partition with no less than
some maximum desired number $k(\textbf q)$ blocks.
\section{Conclusions}
This paper approaches objective function-based fuzzy clustering by firstly eliciting a real-valued cluster score function, quantifying the positive worth of data subsets in the given
classification problem. Clustering is next interpreted in terms of combinatorial optimization via set partitioning. The proposed gradient-based local search relies on a novel
expansion of the MLE of near-Boolean functions \cite{Rossi2015} over the product of $n$ simplices, each of which is $2^{n-1}-1$-dimensional, $n$ being the number of data. The method
needs not the input to specify a desiderd number of clusters, as this latter is determined autonomously through optimization, and applies to any classification problem, handling data
sets not necessarily included in a Euclidean space: proximities between data points and within clusters may be quantified in any conceivable way, including information theoretic measurement
\cite{Pirro+2010}.  

\section{Appendix: continuum of polynomials}
Although not esplicitated, in this work two lattices have appeared thus far, namely the Boolean lattice $(2^N,\cap ,\cup)$ of subsets of $N$ ordered by inclusion $\supseteq$
and the geometric lattice $(\mathcal P^N,\wedge ,\vee)$ of partitions of $N$ ordered by coarsening $\geqslant$ \cite{Aigner79,Stern99}. Both, of course, are posets (partially ordered sets),
and M\"obius inversion applies to any (locally finite) poset, provided a bottom element exists \cite{Rota64}. The bottom subset clearly is $\emptyset$, while the bottom partition
$P_{\bot}=\{\{1\},\ldots ,\{n\}\}$ is the finest one. The analysis provided in this final appendix aims to fully exploit the power of M\"obius inversion towards further detailing and
formalizing an observation appearing in \cite{Rossi2015}, which in turn develops from a result contained in \cite{GilboaLehrer90GG,GilboaLehrer91VI}.

Let $(L,\wedge,\vee)$ be a lattice ordered by $\geqslant$ and with generic elements $x,y,z\in L$. Any lattice function $f:L\rightarrow\mathbb R$ has M\"obius inversion
$\mu^f:L\rightarrow\mathbb R$ given by $\mu^f(x)=\sum_{x_{\bot}\leqslant y\leqslant x}\mu_L(y,x)f(y)$, where $x_{\bot}$ is the bottom element and $\mu_L$ is the M\"obius function, defined
recursively on ordered pairs $(y,x)\in L\times L$ by $\mu_L(y,x)=-\sum_{y\leqslant z<x}\mu_L(z,x)$ if $y<x$ (i.e. $y\leqslant x$ and $y\neq x$) as well as $\mu_L(y,x)=1$ if $y=x$, while
$\mu_L(y,x)=0$ if $y\not\leqslant x$. The M\"obius function of the subset lattice implicitly appears since the beginning of this work, and is $\mu_{2^N}(B,A)=(-1)^{|A\backslash B|}$, with
$B\subset A$. Concerning the M\"obius function of $\mathcal P^N$, given any two partitions $P,Q\in\mathcal P^N$, if $Q<P=\{A_1,\ldots ,A_{|P|}\}$, then
for every block $A\in P$ there are blocks $B_1,\ldots ,B_{k_A}\in Q$ such that $A=B_1\cup\cdots\cup B_{k_A}$, with $k_A>1$ for at least one $A\in P$.
Segment $[Q,P]=\{P':Q\leqslant P'\leqslant P\}$ is thus isomorphic to product $\times_{A\in P}\mathcal P(k_A)$, where $\mathcal P(k)$ denotes the lattice
of partitions of a $k$-set. Accordingly, let $m_k=|\{A:k_A=k\}|$ for $k=1,\ldots ,n$. Then \cite[pp. 359-360]{Rota64},\bigskip

$\mu_{\mathcal P^N}(Q,P)=(-1)^{-n+\sum_{1\leq k\leq n}m_k}\prod_{1<k<n}(k!)^{m_{k+1}}$.\bigskip

Those partition functions $h:\mathcal P^N\rightarrow\mathbb R$ for which there exists a set function $v:2^N\rightarrow\mathbb R$ such that $h(P)=\sum_{A\in P}v(A)$ for all $P\in\mathcal P^N$ may be
said to be additively separable \cite{GilboaLehrer90GG,GilboaLehrer91VI}, with the notation $h=h_v$. An additively separable partition function $h=h_v$ for some $v$ has M\"obius inversion
$\mu^{h_v}$ living only on the modular elements \cite{Stanley1971} of the partition lattice. These are those partitions where only one block, at most, has cardinality $>1$. Therefore, together with
the bottom $P_{\bot}$ and top $P^{\top}=\{N\}$ elements of the lattice, all other modular elements are those partitions of the form $\{A\}\cup P^{A^c}_{\bot}$ for $A\in 2^N$ such that $1<|A|<n$,
where $P_{\bot}^{A^c}$ is the finest partition of $A^c$ \cite[Ex. 13, p. 71]{Aigner79}. The total number of such modular partitions is $2^n-n$. The M\"obius inversion of an additively
separable global game $h_v$ is detailed hereafter (see also \cite[Prop. 4.4, p. 138 and Appendix, p. 144]{GilboaLehrer90GG} and \cite[Prop. 3.3, p. 452]{GilboaLehrer91VI}. 
\begin{proposition}
If $h=h_v$, then $h=h_w$ for a continuuum of set functions $w:2^N\rightarrow\mathbb R,w\neq v$.
\end{proposition}
\begin{proof}
Firstly, by direct substitution,\bigskip

$\mu^{h_v}(P)=\sum_{A\in P}\sum_{B\subseteq A}v(B)\sum_{Q\leqslant P:B\in Q}\mu_{\mathcal P^N}(Q,P)\text{ for all }P\in\mathcal P^N$.\bigskip

Secondly, by the recursive definition of M\"obius function $\mu_{\mathcal P^N}$,\bigskip

$P\neq \{B\}\cup P^{B^c}_{\bot}\Rightarrow\sum_{Q\leqslant P:B\in Q}\mu_{\mathcal P^N}(Q,P)=0$.\bigskip

This entails that M\"obius inversion $\mu^{h_v}$ lives only on modular elements, that is if $P\neq P_{\bot},P^{\top},\{A\}\cup P^{A^c}_{\bot}$ for all $A\in 2^N$,
then $\mu^{h_v}(P)=0$. The $2^n-n$ non-zero values are thus recursively determined as follows: $\mu^{h_v}(P_{\bot})=\sum_{i\in N}v(\{i\})$,
$\mu^{h_v}(\{A\}\cup P_{\bot}^{A^c})=\mu^v(A)$ if $1<|A|<n$, and $\mu^{h_v}(P^{\top})=\mu^v(N)$. Therefore, any $w\neq v$ satisfying
$\sum_{i\in N}v(\{i\})=\sum_{i\in N}w(\{i\})$ and $\mu^v(A)=\mu^w(A)$ for all $A\in 2^N,|A|>1$ also additively separates $h$, that is $h_v=h_w$.
\end{proof}

In view of corollary 8, the setting considered in this work deals precisely with additively separable partition functions, and thus the polynomial expression defined by $\textbf{[2]}$ is not unique.
More specifically, recall that the degree of a polynomial is the highest degree of its terms. Hence in $\textbf{[2]}$, for any chosen set function $w$ additively separating partiton function $h=h_w$,
the degree is $\max\{|A|:\mu^w(A)\neq 0\}$. Furthermore, every non-zero value of M\"obius inversion $\mu^w:2^N\rightarrow\mathbb R$ is a coefficient of the polynomial. It is not hard to see that the
only degree such that there exists a unique set function available for polynomial expression $\textbf{[2]}$ is 0, in which case the only possible set function $w$ takes values $w(\emptyset)=w(A)$ for
all $A\in 2^N$. Indeed, for any degree $k,0<k\leq n$ there exists a continuum of set functions available for additive separability and such that $\max\{|A|:\mu^w(A)\neq 0\}=k$, each defining alternative
but equivalent coefficients of the polynomial. 

\bibliographystyle{abbrv}
\bibliography{biblioPseudoBooleanClustering}

\begin{thebibliography}{10}

\bibitem{Aigner79}
M.~Aigner.
\newblock {\em Combinatorial Theory}.
\newblock Springer, 1997.
\newblock Reprint of the 1979 Edition.

\bibitem{Bensaid++++++96}
A.~Bensaid, L.~Hall, J.~Bezdek, L.~Clarke, M.~Silbiger, J.~Arrington, and
  R.~Murtagh.
\newblock Validity-guided (re)clustering with applications to image
  segmentation.
\newblock {\em IEEE Trans. on Fuzzy Sys.}, 4(2):112--123, 1996.

\bibitem{Bezdek+92}
J.~Bezdek and S.~Pal.
\newblock {\em Fuzzy Models for Pattern Recognition}.
\newblock IEEE Press, 1992.

\bibitem{BorosHammer02}
E.~Boros and P.~Hammer.
\newblock {Pseudo-Boolean optimization}.
\newblock {\em Discrete App. Math.}, 123:155--225, 2002.

\bibitem{HammerNeuralGas}
M.~Cottrell, B.~Hammer, A.~Hasenfu\ss, and T.~Villmann.
\newblock Batch and median neural gas.
\newblock {\em Neural Networks}, 19(6-7):762--771, 2006.

\bibitem{BooleanFunctions}
Y.~Crama and P.~L. Hammer.
\newblock {\em Boolean Functions: Theory, Algorithms, and Applications}.
\newblock Cambridge University Press, 2011.

\bibitem{Du2010}
K.-L. Du.
\newblock Clustering: a neural network approach.
\newblock {\em Neural Networks}, 23:89--107, 2010.

\bibitem{GilboaLehrer90GG}
I.~Gilboa and E.~Lehrer.
\newblock Global games.
\newblock {\em International Journal of Game Theory}, (20):120--147, 1990.

\bibitem{GilboaLehrer91VI}
I.~Gilboa and E.~Lehrer.
\newblock The value of information - an axiomatic approach.
\newblock {\em Journal of Mathematical Economics}, 20(5):443--459, 1991.

\bibitem{Kashef+07}
R.~Kashef and M.~S. Kamel.
\newblock Cooperative clustering.
\newblock {\em Pattern Recognition}, 43(6):2315--2329, 2010.

\bibitem{KorteVygen2002}
B.~Korte and J.~Vygen.
\newblock {\em Combinatorial Optimization. Theory and Algorithms}.
\newblock Springer, 2002.

\bibitem{Krishnapuram+96}
R.~Krishnapuram and J.~Keller.
\newblock The possibilistic c-means algorithm: insights and recommendations.
\newblock {\em IEEE Transactions on Fuzzy Systems}, 4(3):148--158, 1996.

\bibitem{Lughofer08}
E.~Lughofer.
\newblock Extensions of vector quantization for incremental clustering.
\newblock {\em Pattern Recognition}, 41:995--1011, 2008.

\bibitem{Micro}
A.~Mas-Colell, M.~D. Whinston, and J.~R. Green.
\newblock {\em Microeconomic Theory}.
\newblock Oxford University Press, 1995.

\bibitem{Menard+02}
M.~M{\'e}nard and M.~Eboueya.
\newblock Extreme physical information and objective function in fuzzy
  clustering.
\newblock {\em Fuzzy Sets and Systems}, 128:285--303, 2002.

\bibitem{NgSpectral2002}
A.~Y. Ng, M.~I. Jordan, and Y.~Weiss.
\newblock On spectral clustering: analysis and an algorithm.
\newblock In T.~G. Dietterich, S.~Becker, and Z.~Ghahramani, editors, {\em
  Advances in Neural Information Processing Systems 14}, volume~2, pages
  849--856. MIT Press, 2002.

\bibitem{PalBezdek95}
N.~Pal and J.~Bezdek.
\newblock On cluster validity for the fuzzy c-means model.
\newblock {\em IEEE Transactions on Fuzzy Systems}, 3(3):370--379, 1995.

\bibitem{Pardalos++06}
P.~Pardalos, O.~Prokopyev, and S.~Busygin.
\newblock Continuous approaches for solving discrete optimization problems.
\newblock In G.~Appa, L.~Pitsoulis, and H.~Williams, editors, {\em Handbook on
  Modeling for Discrete Optimization}, pages 39--60. Springer, 2006.

\bibitem{Pirro+2010}
G.~Pirr\'o and J.~Euzenat.
\newblock A feature and information theoretic framework for semantic similarity
  and relatedness.
\newblock In {\em Proceedings of The Semantic Web Conference ISWC 2010}, pages
  615--630, 2010.
\newblock LNCS 6496.

\bibitem{Rezaee++98}
M.~Rezaee, B.~Lelieveldt, and J.~Reiber.
\newblock A new cluster validity index for the fuzzy c-means.
\newblock {\em Pattern Recognition Letters}, 19:237--246, 1998.

\bibitem{Rossi2015}
G.~Rossi.
\newblock {\em Continuous set packing and near-Boolean functions}.
\newblock ArXiv, 2015.
\newblock Submitted to ICPRAM 2016.

\bibitem{Rota64}
G.-C. Rota.
\newblock {On the foundations of combinatorial theory I: theory of M{\"o}bius
  functions}.
\newblock {\em Z. Wahrscheinlichkeitsrechnung u. verw. Geb.}, 2:340--368, 1964.

\bibitem{Roubens82}
M.~Roubens.
\newblock Fuzzy clustering algorithms and their cluster validity.
\newblock {\em European Journal of Operational Research}, 10(3):294--301, 1982.

\bibitem{Sandholm02}
T.~Sandholm.
\newblock Algorithm for optimal winner determination in combinatorial auctions.
\newblock {\em Artificial Intelligence}, (135):1--54, 2002.

\bibitem{GraphClustering07}
S.~E. Schaeffer.
\newblock Graph clustering.
\newblock {\em Computer Science Review}, 1(1):27--64, 2007.

\bibitem{KernelMethods2004}
J.~Shawe-Taylor and N.~Cristianini.
\newblock {\em Kernel Method for Pattern Analysis}.
\newblock Cambridge University Press, 2004.

\bibitem{PNAS2005}
N.~Slonim, S.~G. Atwal, G.~Tka$\check{c}$ik, and W.~Bialek.
\newblock Information-based clustering.
\newblock {\em PNAS}, 102(51):18297--18302, 2005.

\bibitem{Stanley1971}
R.~Stanley.
\newblock Modular elements of geometric lattices.
\newblock {\em Algebra Universalis}, (1):214--217, 1971.

\bibitem{Stern99}
M.~Stern.
\newblock {\em Semimodular Lattices. Theory and Applications. Encyclopedia of
  Mathematics and its Applications 73}.
\newblock Cambridge University Press, 1999.

\bibitem{Valente+07}
J.~{Valente de Oliveira} and W.~Pedrycz.
\newblock {\em Advances in fuzzy clustering and its applications}.
\newblock Wiley, 2007.

\bibitem{SpectralClustering08}
U.~von Luxburg, M.~Belkin, and O.~Bousquet.
\newblock Consistency of spectral clustering.
\newblock {\em The Annals of Statistics}, 36(2):555--586, 2008.

\bibitem{WangZhang2007}
W.~Wang and Y.~Zhang.
\newblock On fuzzy cluster validity indices.
\newblock {\em Fuzzy Sets and Systems}, 158:2095--2117, 2007.

\bibitem{WuChowSOM}
S.~Wu and T.~W.~S. Chow.
\newblock Clustering of the self-organizing map using a cluster validity index
  based on inter-cluster and intra-cluster density.
\newblock {\em Pattern Recognition}, 37:175--188, 2004.

\bibitem{XieBeni91}
X.~Xie and G.~Beni.
\newblock Validity measure for fuzzy clustering.
\newblock {\em IEEE Transactions on Pattern Analysis and Machine Intelligence},
  13(8):841--847, 1991.

\bibitem{XuWunsch05}
R.~Xu and D.~Wunsch.
\newblock Survey of clustering algorithms.
\newblock {\em IEEE Trans. on Neural Net.}, 16(3):645--678, 2005.

\bibitem{Zahid+++01}
N.~Zahid, O.~Abouelala, M.~Limouri, and A.~Essaid.
\newblock Fuzzy clustering based on k-nearest-neighbours rule.
\newblock {\em Fuzzy Sets and Sys.}, 120:239--247, 2001.

\end{thebibliography}

\end{document}